\documentclass[sn-basic]{sn-jnl}

\usepackage{float}
\usepackage{latexsym}
\usepackage{amsmath} 
\usepackage{amsfonts}       
\usepackage{amsthm}
\usepackage{geometry}
\usepackage{bbding}         
\usepackage{bm}             
\usepackage{graphicx}    
\usepackage{caption} 
\usepackage[labelfont=bf]{caption}
\usepackage{caption}
\DeclareCaptionLabelSeparator{twospace}{\ ~}   
\usepackage{caption}
\DeclareCaptionLabelSeparator{twospace}{\ ~}   
\usepackage{subcaption}
\usepackage{fancyvrb}       
\usepackage[nottoc]{tocbibind} 
\usepackage{dcolumn}        
\usepackage{booktabs}       
\usepackage{paralist}       

\jyear{2022}

\usepackage{tikz}
\usetikzlibrary{arrows,positioning}
\definecolor{boxcolor}{HTML}{B9DCFF}
\usetikzlibrary{arrows.meta}
\usepackage{subcaption}
\usepackage{float}
\usepackage{multirow}
\usepackage{booktabs}
\usepackage{pgfplots}
\usepackage{bm}
\usepackage{amsfonts}
\usepackage{amsmath, graphicx, epsfig,psfrag, verbatim}
\usepackage{eucal, amssymb,amsmath,amsthm,graphicx, amsfonts, latexsym,xcolor}
\usepackage[utf8]{inputenc}

\newtheorem{lem}{Lemma}[section]

\newtheorem{conj}{Conjecture}[section]
\newtheorem{prop}{Proposition}[section]

\begin{document}

\title[Extending SIRS epidemics to allow for gradual waning of immunity]{Extending SIRS epidemics to allow for gradual waning of immunity}

\author*[1]{\fnm{Mohamed} \sur{El Khalifi}}\email{mohamed.elkhalifi@math.su.se}

\author[1]{\fnm{Tom} \sur{Britton}}\email{tom.britton@math.su.se}

\affil[1]{\orgdiv{Department of Mathematics}, \orgname{Stockholm University}, \orgaddress{\city{Stockholm}, \country{Sweden}}}

\abstract{SIRS epidemic models assume that individual immunity (from infection and vaccination) wanes in one big leap, from complete immunity to complete susceptibility. For many diseases immunity on the contrary wanes gradually, something that's become even more evident during COVID-19 pandemic where also recently infected have a reinfection risk, and where booster vaccines are given to increase immunity. This paper considers an epidemic model allowing for such gradual waning of immunity (either linear or exponential waning) thereby extending SIRS epidemics, and also incorporates vaccination. The two versions for gradual waning of immunity are compared with the classic SIRS epidemic, where the three models are calibrated by having the same \emph{average cumulative immunity}. All models are shown to have identical basic reproduction number $R_0$. However, if no prevention is put in place, the exponential waning model has highest prevalence and the classic SIRS model has lowest. Similarly, the amount of vaccine supply needed to reach and maintain herd immunity is highest for the model with exponential decay of immunity and lowest for the classic SIRS model. consequently, if truth lies close to exponential (or linear) decay of immunity, expressions based on the SIRS epidemic will underestimate the endemic level and the critical vaccine supply will not be sufficient to reach and maintain herd immunity. For parameter choices fitting to COVID-19, the critical amount of vaccine supply is about 50\% higher if immunity wanes linearly, and more than 150\% higher when immunity wanes exponentially, as compared to the classic SIRS epidemic model.}




\maketitle

\section{Introduction}
\label{sec:intro}

When considering infectious disease outbreaks over a longer time horizon, waning of immunity, from disease exposure or vaccination, is known to play an important role. This has been considered in epidemic models for many years, and the most well-studied model is the SIRS (susceptible-infectious-recovered-susceptible) epidemic model, where all individuals are classified as being either susceptible, infectious or recovered (implicitly assuming also being immune), and where individuals eventually loose their immunity and go back to being susceptible after some time, e.g. \citep{hethcote1976qualitative}. The simplest form of this epidemic model, defined by differential equations, assumes that recovered individuals go back to being susceptible at constant rate, thus implying that immunity at the \emph{community level} wanes continuously. However, the SIRS model does \emph{not} allow for partially immune individuals or that immunity wanes gradually at the \emph{individual level}: each individual is either completely immune or fully susceptible.

During Covid-19, but also prior to this, it has become evident that \textit{individual immunity} (to infection) is not a binary property, but rather that individual immunity wanes gradually over time  and can later be boosted either by vaccination or natural infection (see e.g. \citep{goldberg2021waning} for empirical evidence). Quite surprisingly this gradual waning of individual immunity has hence not been yet considered in epidemic models. As a consequence, SIRS epidemic models can never have a group of individuals having lost about half of their immunity, but the models do allow 50\% of the community being completely immune and 50\% being completely susceptible. However, these situations are quite different, in particular when additional individuals get infected.

In the current paper we define and analyse an epidemic model which allows for gradual waning of immunity. This is done by assuming that individuals sequentially loose a portion of their immunity in each step, up to a total of $k$ steps when all immunity is lost. For large $k$ this approximates the situation where immunity drops continuously in time, and we consider both the situation where immunity drops linearly and when immunity drops exponentially (the latter seemingly more biologically reasonable). We call our model the SIR\textsuperscript{(k)}S epidemic model since there now are $k$ immunity (recovered) levels, $k=1$ being the classic SIRS model. It is worth pointing out that the current paper considers immunity to infection,  and not immunity to severe disease and how this wanes. The latter is also an important area which has received attention in several other papers (cf. \cite{hethcote1997age,hethcote1999simulations,carlsson2020modeling}).

The three models, the classical SIRS model with a sudden complete drop of immunity, linear decay of immunity and exponential immunity decay, are calibrated by assuming the same cumulative amount of immunity. So for instance, the SIRS model with, on average, 1 year complete immunity, and then returning to complete susceptibility, is compared with the linear immunity decay model taking two years from complete immunity to ful susceptibility.
For each model we derive expressions for the basic reproduction number $R_0$ and the steady state prevalence (endemic level) if no preventive measures are put in place. We also derive the critical amount of vaccine supply needed to reach and maintain herd immunity, for each of the three models.

Our main conclusion shows that the situation is worse for the more realistic models allowing for gradual waning of immunity compared to the classic SIRS model: even though the three models share the same $R_0$ the models with gradual waning will result in higher prevalence (endemic level) if no preventive measures are put in place, and more vaccine supply (or other preventive measures) are needed to reach a steady herd immunity, implying that vaccination policies (or other preventive measures) based on the SIRS epidemic model may lead to an incorrect sense of security. Among the two studied models for immunity waning, linear and exponential decay, the more realistic exponential decay shows the biggest difference (of endemic prevalence and critical amount of vaccine supply) compared to the classic SIRS model.

\section{Model and main results}
\subsection{Formulation of the models}\label{Sec:ModelFormulation}
All three models assume that a) immunity from vaccination as well as disease exposure initially confer complete immunity, and b) that immunity from vaccination wanes in the same way is immunity from disease exposure. Further, infectious individuals have infectious contacts at rate $\beta$ and recover (and become fully immune) at rate $\gamma$. infect fully susceptible individuals. The differences between the models lie in how immunity wanes, and what is the rate of getting infected for a partially immune in relation to a fully susceptible.

Figure \ref{fig:decay modes} illustrates the immunity waning for the classic SIRS epidemic (assuming waning happens at its expected value) and for the models with linear and exponential decay of waning.

\begin{figure}[h!]
    \includegraphics[width=.7\linewidth]{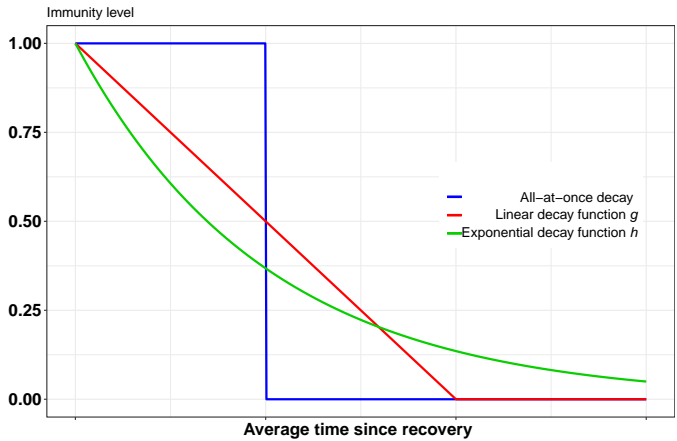} 
    \centering
  \caption{Different modes of decay of immunity on individual level: (blue) all-at-once decay (taking place at its expected time-point), (red) linear decay function $g(u) = 1-\omega/2 u$ if $0\leq u\leq 2/\omega$, and $g(u)=0$ otherwise, and (green) exponential decay function $h(u) = \exp \left(-\omega u\right), u\geq 0$. All three models having the same average cumulative immunity.}
  \label{fig:decay modes}
\end{figure}

\textbf{The classic SIRS epidemic}. The classic SIRS epidemic model assumes that immunity drops from complete immunity to complete susceptibility in one single step at a constant rate $\omega$ (so the mean duration of full immunity equals $\omega^{-1}$) \citep{hethcote1976qualitative}. The model is illustrated in Fig. \ref{fig:SIRS}, where $ s(t) $, $ i(t) $ and $ r(t) $ denote community fractions of susceptible, infectious, and recovered (=immune) individuals at time $ t $, respectively. The model is defined by the following three differential equations:
\begin{equation}\label{SIRS}
	\begin{array}{lll}
		s'(t) &=& \mu - \beta s(t) i(t) + \omega r(t) - \mu s(t),\\
		i'(t) &=& \beta s(t) i(t) - (\gamma+\mu) i(t),\\
		r'(t) &=& \gamma i(t) - (\omega + \mu) r(t).
	\end{array}
\end{equation}

\textbf{The classic SIR\textsuperscript{(k)}S epidemic with linear/exponential waning}. Our new model, denoted the SIR\textsuperscript{(k)}S epidemic model, instead assumes that immunity wanes sequentially in $k$ steps (for some large $k$), as illustrated in Fig. \ref{fig:SIRkS} ($k=1$ gives the classic SIRS epidemic). The linear version does so by choosing the $k$ down-jumps and their corresponding rates such that the decay mimics a linear decay, and the exponential version chooses down-jumps and rates to mimic exponential decay, and both models do this in a way such that the cumulative immunity equals $\omega^{-1}$ (independent of $k$) just like the SIRS model. The new model is illustrated in Figure \ref{fig:SIRkS} and defined in detail with $k+2$ differential equations in Section \emph{Materials and Methods}. There $r_0(t)$ denotes the community fraction being having no susceptibility, $r_1(t)$ the community fraction having gained one level of susceptibility, and so on, and $r_{k-1}(t)$ the fraction having susceptibility level $k-1$ being the last step before becoming completely susceptible.

In \citep{machlaurin2020cost}, similar waning functions were used to model the vaccination efficacy over time while they estimate the cost-effective vaccination strategy against tuberculosis. 

\begin{figure}[htbp]
\centering
\begin{subfigure}[b]{.5\linewidth}
\centering
\resizebox{\columnwidth}{!}{%
   \begin{tikzpicture}[node distance=1cm, auto,
    >=Latex, 
    every node/.append style={align=center},
    int/.style={draw, minimum size=1cm}]

   \node [int] (S)             {$s(t)$};
   \node [int, right=of S] (I) {$i(t)$};
   \node [color=green,very thick,int, right=of I] (R) {$r(t)$};

   \coordinate[right=of I] (out);
   \path[->] (S) edge[thick] node {$\beta i(t)$} (I)
             (I) edge node {$\gamma$} (R)
               (R.north)  edge[out=120, in=60] node {$\omega$} (S.north);
\end{tikzpicture}
}
   \caption{}
   \label{fig:SIRS} 
\end{subfigure}
\begin{subfigure}[b]{.7\linewidth}
\resizebox{\columnwidth}{!}{%
   \begin{tikzpicture}[node distance=1cm, auto,
    >=Latex, 
    every node/.append style={align=center},
    int/.style={draw, minimum size=1cm}]
\centering
   \node [int] (S)             {$s(t)$};
   \node [int, right=of S] (I) {$i(t)$};
   \node [color=green,very thick,int, right=of I] (R0) {$r_0(t)$};
   \node [color=green,very thick,int, right=of R0] (R1) {$r_1(t)$};
   \node [color=green,very thick,int, right= 2cm of R1] (Rk1) {$r_{k-1}(t)$};
   
   \draw [dashed,->] (R1) -- node [above] {}  (Rk1);
   
   \coordinate[right=of I] (out);
   \path[->] (S) edge[thick] node {$\beta i(t)$} (I)
             (I) edge node {$\gamma$} (R0)
             (R0) edge node {$c_k(1)$} (R1)
               (Rk1.south)  edge[out=-120, in=-130, thick] node[below] {$\beta \frac{k-1}{k} i(t)$} (I.west)
               (Rk1.north)  edge[out=120, in=60] node {$c_k(k)$} (S.north)
               (R1.south) edge[out=-120, in=-130, thick] node[above] {$\beta \frac{1}{k} i(t)$} (I.west);
\end{tikzpicture}
}
\caption{}
\label{fig:SIRkS}
\end{subfigure}
\caption{(a) Diagram of the standard SIRS epidemic model. (b) Diagram of the SIR\textsuperscript{(k)}S epidemic model. The green boxes represent the different classes of partially immune states.}
\end{figure}
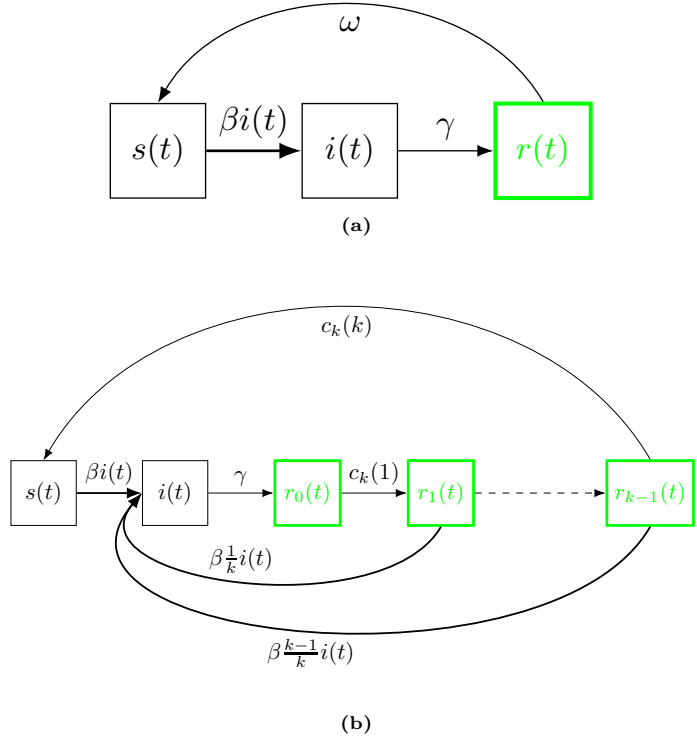

We let SIR\textsuperscript{($\infty$)}S denote the model being the limit of the  SIR\textsuperscript{(k)}S model as $k$ goes to $\infty$ (in our illustrations we use $k=1000$). This limiting model converges to an ODE-PDE system with three equations \citep{kermack1932contributions}, see (Section \ref{Sec.PDE}).

\textbf{Introducing vaccination.} As mentioned earlier, we assume that vaccine as well as infection initially give full immunity, and that the two immunities wane in the same way. 

In the classic SIRS model each individual is either fully susceptible or completely immune at any given point in time, and if this immune status is known it of course only makes sense to vaccinate among the fully susceptible individuals at some rate $\eta$ (why waste vaccines on fully immune?). 

In the case where immunity wanes continuously,  vaccines can in principle be distributed in many different ways (Fig. \ref{fig:generalvacc}). However, since individuals only differ in terms of susceptibility in our model, it should be clear that the class of rational vaccination strategies consist of vaccinating individuals as soon as their immunity drops below some fixed level $\iota$ (or equivalently when the susceptible reaches the level $1-\iota$). The level $\iota$ will determine how much vaccine that will be required: the larger $\iota$  the bigger vaccine supply $\theta$ is needed. For finite $k$ this amounts to vaccinate fully susceptibles at rate $\eta_s^\star$, to not vaccinate in states $r_0$ up to $r_{j-1}$ for some $j\in\{1,\cdots,k-1\}$, to vaccinate $r_j$ at some rate $\eta_j^\star$, and to immediately vaccinate individuals who go from state $r_j$ to $r_{j+1}$ (so the fractions in those states will equal 0). Since individuals in state $r_j$ who loose more immunity are immediately vaccinated, the effective vaccination rate equals in this classe is $\eta_j^\star+c_k(j+1)$.
Fig. \ref{fig:rationalvacc} represents the corresponding SIR\textsuperscript{(k)}S model with such vaccination scheme. An important question is hence to determine how much vaccine supply $\theta_c$ (critical vaccine supply) is needed to reach and maintain herd immunity.

\begin{figure}[htbp]
\centering
\begin{subfigure}[b]{.7\linewidth}
\centering
\resizebox{\columnwidth}{!}{%
\begin{tikzpicture}[node distance=1cm, auto,
    >=Latex, 
    every node/.append style={align=center},
    int/.style={draw, minimum size=1cm}]
   \node [int] (S)             {$s$};
   \node [int, right=of S] (I) {$i$};
   \node [color=green,very thick,int, right=of I] (R0) {$r_0$};
   \node [color=green,very thick,int, right=of R0] (R1) {$r_1$};
   \node [color=green,very thick,int, right= 2cm of R1] (Rk1) {$r_{k-1}$};
   \draw [dashed,->] (R1) -- node [above] {}  (Rk1);
   
   \coordinate[right=of I] (out);
   \path[->] (S) edge[thick] node {$\beta i(t)$} (I)
             (I) edge node {$\gamma$} (R0)
             (R0) edge node {$c_k(1)$} (R1)
               (Rk1.south)  edge[out=-120, in=-130, thick] node[below] {$\beta \frac{k-1}{k} i(t)$} (I.west)
               (Rk1.north)  edge[out=120, in=60] node {$c_{k}(k)$} (S.north)
               (R1.south) edge[out=-120, in=-130, thick] node[above] {$\beta \frac{1}{k} i(t)$} (I.west);
\path[->] (S.north) edge[color=blue, out=60, in=120] node {$\eta_s$} (R0.north)
           (R1.north) edge[color=blue, out=120, in=60] node {$\eta_1$} (R0.north)
           (Rk1.north) edge[color=blue, out=120, in=60] node {$\eta_{k-1}$} (R0.north);
          
\end{tikzpicture}
}
\caption{}\label{fig:generalvacc}
\end{subfigure}
\begin{subfigure}[b]{.7\linewidth}
\centering
\resizebox{\columnwidth}{!}{%
\begin{tikzpicture}[node distance=1cm, auto,
    >=Latex, 
    every node/.append style={align=center},
    int/.style={draw, minimum size=1cm}]
   \node [int] (S)             {$s$};
   \node [int, right=of S] (I) {$i$};
   \node [color=green,very thick,int, right=of I] (R0) {$r_0$};
   \node [color=green,very thick,int, right= 2cm of R0] (Rj) {$r_{j}$};
   \draw [dashed,->] (R0) -- node [above] {}  (Rj);
   \coordinate[right=of I] (out);
   \path[->] (S) edge[thick] node {$\beta i(t)$} (I)
             (I) edge node {$\gamma$} (R0)
              (Rj.south)  edge[out=-120, in=-130, thick] node[below] {$\beta \frac{j-1}{k} i(t)$} (I.west);
\path[->] (S.north) edge[color=blue, out=60, in=120] node {$\eta_s^\star $} (R0.north)
           (Rj.north) edge[color=blue, out=120, in=60] node [swap] {$\eta_{j}^\star+ c_k(j+1)$ } (R0.north);
\end{tikzpicture}
}
\caption{}\label{fig:rationalvacc}
\end{subfigure}

\caption{Diagram of SIR\textsuperscript{(k)}S epidemic model with vaccination. 
(a) General vaccination scheme where partially susceptible individuals $r_j$ are vaccinated at rate $\eta_j$ for any $j=1,\cdots,k-1,$ respectively. (b) Rational vaccination scheme where fully susceptible individuals are vaccinated at rate $\eta_s^\star$ and, if needed, only one class $r_{j}$ of partially susceptible individuals is vaccinated at rate $\eta_{j}^\star + c_k(j+1) $ for some $j\in\{1,\cdots,k-1\}$.}
\label{fig:SIRkS vaccination}
\end{figure}
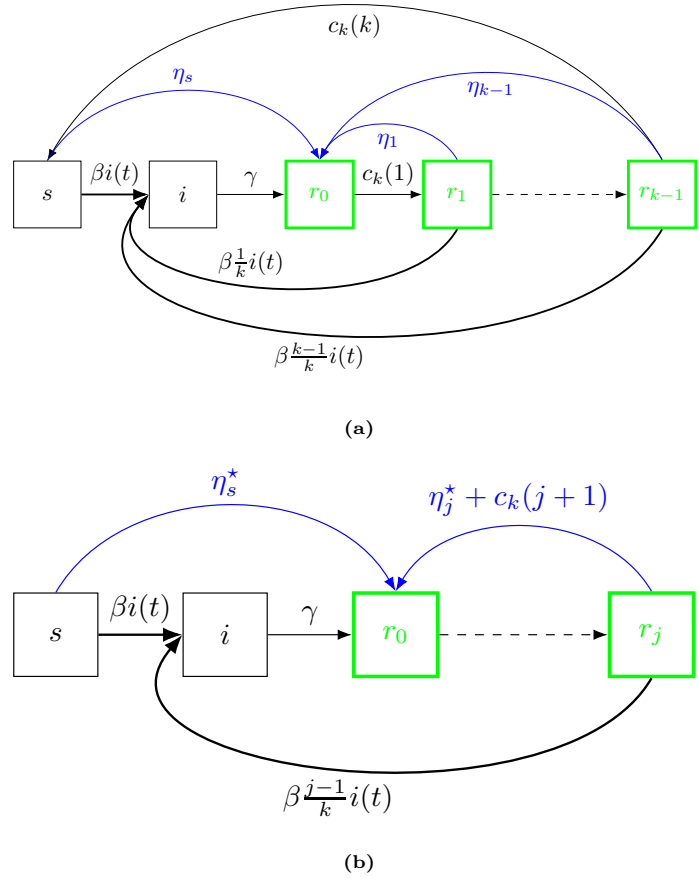

\subsubsection{Parameter choice} In what follows we will compare the three models in terms of steady state prevalence (endemic level) if no preventions are put in place, and how much vaccine that is required to reach and maintain herd immunity. 
In Table \ref{table:Params-values} we show the mid value and range for the model parameters that are used in our \emph{Results} Section when comparing the classic SIRS model with our new models having linear and exponential waning of immunity respectively. Those values are commonly used to characterize diseases like Covid-19, influenza, common cold, etc \citep{byrne2020inferred,davies2020age,hall2022protection,zhang2020changes}. Average life expectancy is set at 80 years. Please note that the results are hence not applicable to childhood diseases (measles, chickenpox, ...) where immunity typically is close to life-long. 
\begin{table}[ht]
	\renewcommand{\arraystretch}{1.5}
	\centering
	\caption{Parameters description, their baseline values, and ranges of variation studied.}
	\label{table:Params-values}
	\begin{tabular}[t]{lp{3cm}cc}
		\toprule
		Parameter &Description& Baseline value&Range\\
		\midrule
		$ R_0 $& Basic reproduction number&5&1--7\\
		$ \gamma^{-1} $ &Mean infectious period (in days) &7&3--14\\
		$ \omega^{-1} $ &Average immune period (in months) &12 &6--24\\
		\bottomrule
	\end{tabular}
\end{table}

\subsection{Main results}\label{sec-Results}
We now compare the three epidemic models, the classic SIRS, the model with linearly waning of immunity, and the model with exponentially decaying immunity, all three models being calibrated by having the same cumulative immunity. Analytical results are obtained for $k=2$ (Appendix \ref{appendix1}) and conjectured to any $k>2.$

\subsubsection{The basic reproduction number $R_0$}
The basic reproduction number, defined as the number of secondary cases produced by one infectious individual in a fully susceptible population, equals $R_0 = \frac{\beta}{\gamma+\mu}$ for the classic SIRS model as well as our extended models. This holds true because the models differ only in terms of how immunity wanes and in the initial phase of epidemic when nearly everyone is susceptible immunity waning has no impact. From now on we assume that $R_0>1$ -- otherwise none of the three models will experience any outbreak and vaccination is not necessary.

\subsubsection{Long-term prevalence in the absence of vaccination}
A comparison of long-term prevalence is obtained by setting the defining differential equations for each of the three models (given in Section \emph{Materials and Methods}) equal to 0 and solving the equation system. When $R_0>1$ there is one stable solution with a positive fraction infectives $\hat i$, the endemic level or stable prevalence. In Fig. \ref{fig:EndLevelandRC} these endemic levels are given for the three models as a function of $R_0 $ (keeping the mean infectious period and average cumulative immunity fixed). It can be seen that the model linear waning of immunity results in larger endemic levels than the SIRS epidemic. The model with exponential waning of immunity makes the long-term prevalence even larger. When $R_0\approx5$ as for Covid-19 Delta strain (e.g. \cite{zhang2020changes}) and with a mean infectious period of 7 days and an average duration of immunity of 1 year, the stable prevalence will consist of 1.6\% being infectious according to the SIRS model. The linear waning model has about twice the endemic level (3\% of the population) and the model with exponential waning has stable prevalence 4.9\% (Fig. \ref{fig:EndLevelandRC}).

\begin{figure}[htbp]
\includegraphics[width=.7\linewidth]{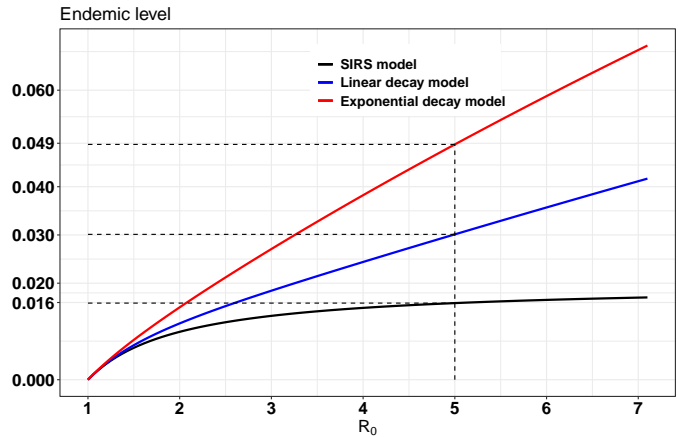} 
\centering
\caption{Endemic levels from the standard model and the SIR\textsuperscript{(k)}S model with linear and exponential decay functions.}
\label{fig:EndLevelandRC}
\end{figure}

\subsubsection{Critical vaccine to reach and maintain herd immunity}
In Fig. \ref{fig:VaccineSupplyR0} we show the necessary amount of vaccine supply (for the three models) continuously needed to reach and maintain herd immunity (see \emph{Materials and Methods} for the derivation). 
It is seen that the standard SIRS model requires a lower vaccine supply as compared to the two models with gradual waning, and that the model with exponential immunity waning require the largest vaccine supply. 
 Moreover, the difference between the three models grows with $R_0$. Take as illustration $R_0\approx5$ (and mean infectious period 7 days and cumulative immunity 1 year (inspired by COVID-19 pandemic Delta strain), then the classic SIRS model requires vaccinating at rate 0.81 to reach herd immunity (so 8.1 million vaccinations per year in a population of 10 million), the model with linear waning requires 1.25, and the exponential decay model requires 2.14, some 55\% and 164\% more vaccinations, respectively.
\begin{figure}[htbp]
    \includegraphics[width=.7\linewidth]{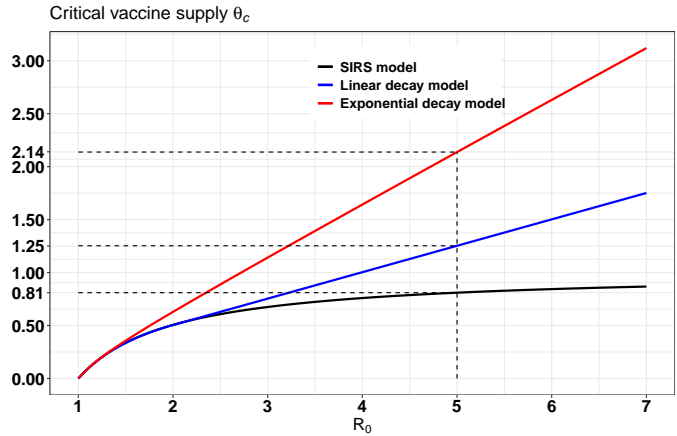} 
    \centering
  \caption{Critical amount of vaccines needed to reach herd immunity. }
  \label{fig:VaccineSupplyR0}
\end{figure}

\subsubsection{Comparison of the models for endemic diseases}
For new emerging diseases $R_0$ is often estimated from the initial growth rate of the epidemic (together with knowledge about the generation time of the disease) \citep{hu2021infectivity}. Then the natural calibration of models was to assume the same cumulative immunity $\omega^{-1}$, and the same transmission rate $\beta$, recovery rate $\gamma$  as done above. 

For diseases that are currently endemic, a more natural calibration is instead to assume the different models have the same cumulative immunity $\omega^{-1}$ and the same recovery rate $\gamma$, and that the endemic level equals the empirical level (so fixing the endemic level  rather than $R_0$). An argument for this calibration is that, while immunity duration and infectious period may be easy to estimate, the same is usually not true for the rate of infectious contacts $\beta$, which in turns determines $R_0$. Fig \ref{fig:R0 comparison} shows the estimated $R_0$ for the different models based on such a calibration, for different values of the endemic level (stable prevalence). The estimate based on the SIRS epidemic was derived in \citep{heffernan2005perspectives}.

\begin{figure}[htbp]
\centering
  \begin{subfigure}{.7\linewidth}
    \centering
    \includegraphics[width=\linewidth]{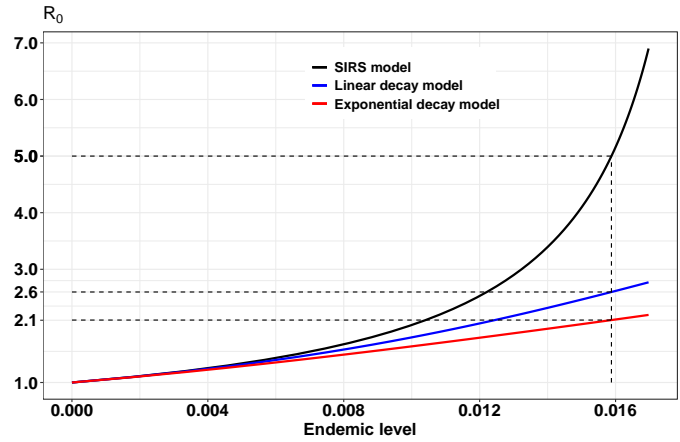}
    \caption{}
  \label{fig:R0 comparison}
  \end{subfigure}
  \begin{subfigure}{.7\linewidth}
    \centering
    \includegraphics[width=\linewidth]{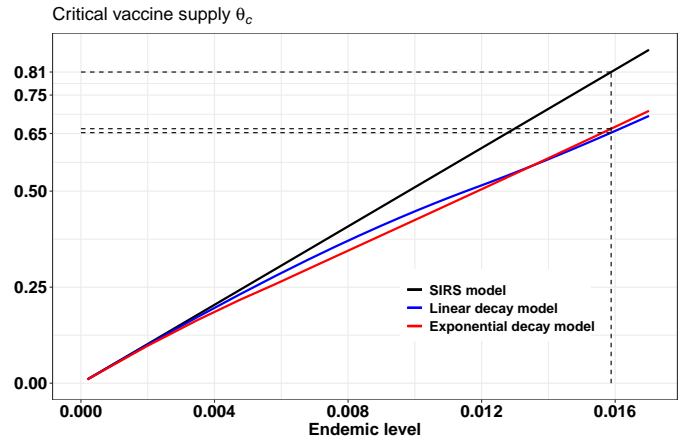}
    \caption{}
    \label{fig:VaccineSupplyEnd}
  \end{subfigure}
  \caption{ (a) $R_0$ estimates from prevalence (endemic) data and (b) the corresponding amount of vaccines needed to reach herd immunity.}
  \label{fig:EndLevel and RC}
\end{figure}

As seen in the figure, the SIRS model results estimates $R_0$ being larger than the models with gradual waning, and in particular compared to the model with exponential waning.  When the endemic level (prevalence) is 1.6\% of the population, the exponential decay model and the linear decay model estimate $R_0$ to 2.1 and 2.6, respectively, which are 58\% and 48\% less than the value 5 estimated by the classic SIRS model. Yet, when the prevalence is low, all models result in approximately the same $R_0$.

If we instead use the endemic level (together with fixing cumulative immunity and the infectious period) to estimate the required amount of vaccine supply, for the different models, then the classic SIRS model requires higher vaccines. This is because the disease is spreading faster in the SIRS setting than in the models with gradual waning. This is illustrated in Fig. \ref{fig:VaccineSupplyEnd} which shows that slightly less vaccine supply is required to reach herd immunity under the SIR\textsuperscript{($\infty$)}S models in comparison to the standard SIRS model, and this applies for both linear and exponential decay modes of immunity (19\% less vaccines if the disease persists in 1.6\% of the population).
\section{Discussion}
The classic SIRS model assumes that immunity at the individual level is binary, i.e.\ each individual is either fully immune or fully susceptible. This paper relaxes this assumption by presenting and analyzing a novel model allowing for gradual waning of immunity, either linear waning or exponential waning. It is shown that, when calibrating the models by assuming the same $R_0$ and mean infectious period and cumulative infectivity, the new more realistic models result in higher endemic levels if prevention is not put in place, and that a substantially larger vaccine supply is required to reach and maintain herd immunity. The most realistic model having exponential waning of immunity is shown to exhibit the biggest difference between the classic SIRS model.

The studied model can in principle be defined also for other forms of deterministic immunity waning modes. 

Our model extends the SIRS epidemic to allow for gradual waning of immunity. Many other model assumptions are admittedly unrealistic. It would of course be interesting to study this extension to gradual waning of immunity when also allowing for e.g.\ not obtaining full immunity from start, having different forms of immunity for vaccine as compared to disease exposure, considering asymptomatic and symptomatic individuals. Still it our belief that the qualitative feature, that gradual waning requires bigger vaccine supply, remains.

Another assumption was that the immunity status of individuals were known when determining whom to vaccinate. In case immunity wanes deterministically, as in the two new models, this might be a reasonable approximation since the time of last vaccination or infection might be known, but when immunity wanes in one leap after an exponential time, this may not be possible. Analysing models where the exact immunity status, perhaps also introducing randomness in waning decay, is an interesting problem to analyse.

\section{Material and Methods}
We summarize the methods used to establish the results listed in Section \emph{Results}. We start by formulating the SI$\mbox{R}^{(k)}$S model with gradual waning of immunity including linear and exponential decaying functions. A rigorous mathematical analysis of the model is given in the case $k=2$ (Appendix \ref{appendix1}), thus, allowing to make conjectures for any $k>1.$
\subsection{General SI$\mbox{R}^{(k)}$S epidemic model}\label{Sec:general SIRkS model}
The general SI$\mbox{R}^{(k)}$S model (Fig. \ref{fig:SIRkS}) we introduced in this paper aims to approximate the linear and exponential immunity decay modes (Fig. \ref{fig:decay modes}) using step functions such that all immunity is lost in $k$ jumps, starting by complete immunity to zero immunity, loosing a portion $\frac{1}{k}$ each step. Here we outline how to construct the SI$\mbox{R}^{(k)}$S model following any function of waning of immunity.

Suppose that for a given decaying function and an arbitrary integer $k\geq 1$, immunity level $ \frac{k-j}{k}$ lasts for an exponentially distributed time with rate $c_k(j+1)>0$ before decaying to $ \frac{k-j}{k}-\frac{1}{k}$ with $j$ going from 0 to $k-1$, such that the rates $  \{ c_k(j)\}_{j=1}^k $ verify the constant cumulative immunity condition
\begin{equation}\label{Susc.lev.cumul.immunity}
		\sum\limits_{j=0}^{k-1} \frac{1}{c_k(j+1)}\left(1- \frac{j}{k}\right)  = \frac{1}{\omega},
\end{equation}
and approximate the underlying waning of immunity. Denote by $\{r_j(t)\}_{j=0}^{k-1} $ the fractions of individuals (at time $t$) with the immunity level $\frac{k-j}{k}$. Clearly, recovered infectious individuals enter to the highest immunity class, $r_0(t)$, and then their immunity declines through $k$ steps.  The class of individuals $r_{k-1}(t)$ has the lowest immunity level, $\frac{1}{k}$, to be lost altogether to become fully susceptible again.
Thus, the resulting SI$\mbox{R}^{(k)}$S model with gradual decay of immunity can be formulated as in the following equation
\begin{equation}\label{SIRkS}
	\begin{array}{lll}
		s'(t) &=& \mu - \beta s(t) i(t) + c_{k}(k) r_{k-1}(t) - \mu s(t),\\
		i'(t) &=& \beta s(t) i(t) + \beta \sum\limits_{j=1}^{k-1}\frac{j}{k} r_{j}(t) i(t) - (\gamma+\mu) i(t),\\
		r_0'(t) &=& \gamma i(t) - (c_{k}(1)+\mu) r_0(t),\\
		r_j'(t) &=& c_{k}(j) r_{j-1}(t)  - \beta \frac{j}{k}r_{j}(t)i(t) - (c_k(j+1)+\mu) r_j(t),
	\end{array}
\end{equation}
for $j=1,\cdots,k-1$, where we omit the dependence on $k$ in $(s(t), i(t),r_{0}(t),\cdots, r_{k-1}(t))$ for simplicity of notation. We will also use the notation $(\hat s, 0, \hat r_{0},\cdots, \hat r_{k-1})$ for the disease free equilibrium and $(\bar s, \bar i, \bar r_{0},\cdots, \bar r_{k-1})$ for the endemic equilibrium.

The disease free equilibrium of \eqref{SIRkS} is $E^{k,0}=(1, 0, \cdots
,0)$ and the basic reproduction number is given by $$ R_0 = \dfrac{\beta}{\gamma+\mu}. $$

The sequence $ \{ c_k(j)\}_{j=1}^k $ is chosen to fit the required decay mode of immunity and the fixed cumulative immunity condition \eqref{Susc.lev.cumul.immunity} regardless of $k$.

\subsection{SI$\mbox{R}^{(k)}$S models with linear and exponential waning of immunity}\label{lin and exp rates}
The linear and the exponential functions modelling the decay of immunity (Fig. \ref{fig:decay modes}) given by
\begin{equation}
g(u) = \left(1-\frac{\omega}{2}a\right){\boldsymbol 1}_{\{a<\frac{2}{\omega}\}} \, \mbox{ and } h(u) = \exp\left( -\omega u\right),\, u\geq 0,
\end{equation}
respectively, with the indicator function ${\boldsymbol 1}_A $ equals 1 if the condition $A$ holds and 0 otherwise, verify the same cumulative immunity condition
\begin{align*}
    \int_{\mathbb{R}^+}g(u)\, du = \int_{\mathbb{R}^+}h(u)\, du = \frac{1}{\omega},
\end{align*}
which is equal to the average cumulative immunity from the standard SIRS model with immunity waning rate $\omega$. One can fit the linear decay mode $g$ by letting
\begin{align}\label{lin_rates}
    c_k(j) =\frac{k+1}{2}\omega, \,\,j=1,\cdots,k,
\end{align}
which verifies the condition \eqref{Susc.lev.cumul.immunity}. We refer to the model \eqref{SIRkS} with \eqref{lin_rates} as the SI$\mbox{R}^{(k)}$S model with linear decay of immunity. A way to fit the exponential decay mode is to choose, for any $j=1,\cdots,k-1,$
\begin{align}\label{exp_rates}
c_k(j) = \left( -\frac{1}{\omega}\log(1- j(k-1)/k^2  ) - \sum\limits_{l=1}^{j-1} \frac{1}{c_k(l)}\right)^{-1},
\end{align}
and obtain $c_k(k)$ by solving the equation \eqref{Susc.lev.cumul.immunity}. We refer to the resulting model as the SI$\mbox{R}^{(k)}$S model with exponential decay of immunity. Fig. \ref{fig:StepFunctions} plots the corresponding step functions for $k=10$ and where the duration of each immunity level is set to its expected value.
\begin{figure}[htbp]
    \includegraphics[width=.7\linewidth]{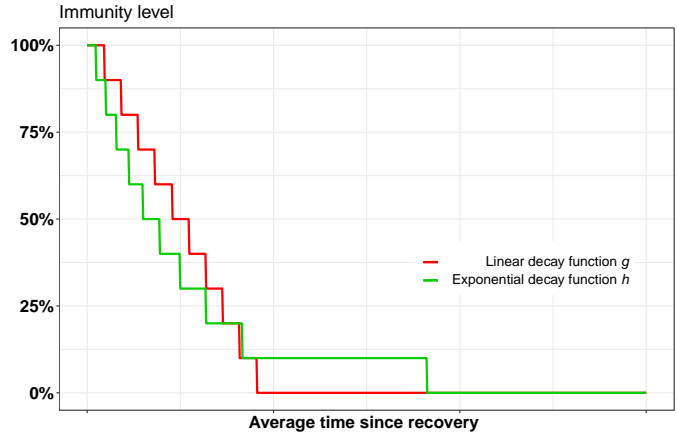} 
    \centering
  \caption{Step functions approximating the immunity decay functions $g$ and $h$ in $k=10$ steps. }
  \label{fig:StepFunctions}
\end{figure}

Different rates $  \{ c_k(j)\}_{j=1}^k $ other than \eqref{lin_rates} and \eqref{exp_rates} can be considered to approximate the linear and the exponential decay modes, respectively. Still, they will have no effect on the model dynamics as $k\rightarrow\infty$.

\subsection{SI$\mbox{R}^{(k)}$S model with vaccination}\label{Sec.Vacc_METHOD}
We now introduce vaccination into the model and make the simplifying that immunity from vaccination is identical to immunity from disease exposure (complete immunity with the same decaying mode). The resulting SI$\mbox{R}^{(k)}$S model with a general vaccination scheme can be written as
\begin{equation}\label{SIRkS vaccination}
	\begin{array}{lll}
		 s^{\boldsymbol \eta}{}' (t) =& \mu - \beta s^{\boldsymbol \eta}(t) i^{\boldsymbol \eta}(t) + c_{k}(k) r_{k-1}^{\boldsymbol \eta}(t) - \mu s^{\boldsymbol \eta}(t) - \eta_s s^{\boldsymbol \eta}(t),\\
		i^{\boldsymbol{\eta}}{}' (t) =& \beta s^{\boldsymbol{\eta}}(t) i^{\boldsymbol{\eta}}(t) + \beta \sum\limits_{j=1}^{k-1}\frac{j}{k} r_{j}^{\boldsymbol{\eta}}(t) i^{\boldsymbol{\eta}}(t) - (\gamma+\mu) i^{\boldsymbol{\eta}}(t),\\
		r_0^{\boldsymbol{\eta}}{}'(t) =& \eta_s s^{\boldsymbol{\eta}}(t) + \gamma i^{\boldsymbol{\eta}}(t) - (c_{k}(1)+\mu) r_0^{\boldsymbol{\eta}}(t)+ \sum\limits_{j=1}^{k-1} \eta_j r_j^{\boldsymbol{\eta}}(t),\\
		r_j^{\boldsymbol{\eta}}{}'(t) =& c_{k}(j) r_{j-1}^{\boldsymbol{\eta}}(t)  - \beta \frac{j}{k}r_{j}^{\boldsymbol{\eta}}(t)i^{\boldsymbol{\eta}}(t) - (c_k(j+1)+\mu) r_j^{\boldsymbol{\eta}}(t) \\ 
		&- \eta_j r_j^{\boldsymbol{\eta}}(t),
	\end{array}
\end{equation}
for $j=1,\cdots,k-1,$ where $\eta_s$ and $\eta_j\geq 0,\,j=1,\cdots,k-1$ are the rates of vaccination of $s^{\boldsymbol{\eta}}(t)$ and $r_j^{\boldsymbol{\eta}}(t),\,j=1,\cdots,k-1$, respectively. The disease free equilibrium $E_v^{k,0}=(\hat{s}^{\boldsymbol{\eta}}, 0, \hat{r}_{0}^{\boldsymbol{\eta}}, \cdots, \hat{r}_{k-1}^{\boldsymbol{\eta}})$ is given by
\begin{align}
\hat{s}^{\boldsymbol{\eta}} &= \frac{\mu}{\mu+\eta_s - \eta_s c_k(k)A_k B_k  }, \\
\hat{r}_{j}^{\boldsymbol{\eta}} &= \eta_s \hat{s}^{\boldsymbol{\eta}} A_k B_j , \,\,j=1,\cdots,k-1,\\
\hat{r}_{0}^{\boldsymbol{\eta}} &=  1 - \hat{s}^{\boldsymbol{\eta}} - \sum\limits_{j=1}^{k-1} \hat{r}_{j}^{\boldsymbol{\eta}},
\end{align}
where $B_j = \prod\limits_{l=1}^{j-1}\frac{c_k(l)}{\mu+c_k(l+1)+\eta_l}, j=1,\cdots,k,$ and 
    $A_k= \left( \mu+c_k(1)- \sum\limits_{j=1}^{k-1}\eta_j B_j \right)^{-1}. $
The effective reproduction number is given by 
\begin{align}
R_e^{(k)}=R_0\left( \hat{s}^{\boldsymbol{\eta}} + \sum\limits_{j=1}^{k-1} \frac{j}{k} \hat{r}_{j}^{\boldsymbol{\eta}} \right),
\end{align}
where we recall that $\frac{j}{k} $ is the relative susceptibility in the $j$'th immunity state.
\subsection{Critical vaccine supply}
The vaccine usage (per unit of time) for the general vaccination scheme of the previous subsection, once it has reached steady state, is given by
\begin{equation}
\theta^{(k)}= \eta_s\hat{s}^{\boldsymbol{\eta}} + \sum\limits_{j=1}^{k-1} \eta_j\hat{r}_{j}^{\boldsymbol{\eta}}.
\end{equation}
For fixed $k,$ the best vaccination strategy, given some amount of vaccine supply delivered continuously, is clearly to vaccinate the most susceptible (=least immune) individuals. More precisely, the best strategy is to immediately vaccinate individuals having higher susceptibility than some class $j$, to vaccinate individuals in susceptibility class $j$ at rate $\eta_j^\star$, and to not vaccinate individuals in susceptibility classes lower than $j$ (i.e.\ $r_0,\dots r_{j-1}$), where $j$ and $\eta_j^\star$ will depend on the amount of available vaccine. With this vaccination strategy, individuals moving to state $j+1$ (with rate $c_k(j+1)$) will be vaccinated immediately and no individuals will ever reach higher susceptibility classes, so the actual vaccination rate among individuals in class $r_j$ is $\eta_j^\star+c_k(j+1)$. Newborns should also be vaccinated, at their incoming rate $\mu$, when we vaccinate in class the class $r_j$. For this strategy to be successful in the long run the amount of vaccine supply should be such that the corresponding $j$ and $\eta_j^\star$ result in $ R^{\boldsymbol{\eta}}\leq 1$. Hence, the critical vaccine supply can be written as
\begin{equation}\label{criticalv}
\theta_c^{(k)}=
\begin{cases}
\eta_s^\star\hat{s}^{\boldsymbol{\eta}} \qquad\qquad\,\,\mbox{ if we only vaccinate fully susceptibles,}\\
\mu + (\eta_j^\star+c_k(j+1))\hat{r}_{j}^{\boldsymbol{\eta}}\, \mbox{ if we vaccinate in the class $r_j$}.
\end{cases}
\end{equation}
A detailed derivation of \eqref{criticalv} when $k=2$ is given in Proposition \ref{proposition}.

An alternative way to derive the critical vaccine supply is to assume that the disease is in the endemic steady state and then vaccinate in each immunity class and to check that the disease-free equilibrium $E_v^{k,0}$ is the only stable steady state.
\subsection{Standard SIRS epidemic model}
It has been shown that the standard SIRS model \eqref{SIRS} admits a unique endemic equilibrium when $R_0>1$ and only the disease-free equilibrium exists when $R_0\leq 1$
\citep{hethcote1976qualitative}. When susceptibles are vaccinated at a constant rate $\eta$, the resulting SIRS model has a unique endemic equilibrium when $R^{\boldsymbol{\eta}}>1$, and only the disease-free equilibrium exists when $R^{\boldsymbol{\eta}}\leq 1$, where $R^{\boldsymbol{\eta}} = R_0 \hat{s}$ is the average number of new infections generated by an infective individual in a population with susceptible fraction of $\hat{s}$ \citep{hethcote1978immunization}. Since $ \hat{s} = \frac{\mu+\omega}{\mu+\omega+\eta} $, the minimum vaccination rate to drive the epidemic dynamic to the disease-free state (i.e., at which $R_0 \hat{s}=1$) is given by $$\eta_c = (\omega+\mu)(R_0-1).$$
Hence, the critical vaccine supply required to achieve and maintain the disease free equilibrium is defined as the product of the rate $\eta_c$ and the fraction $\hat{s}=1/R_0$ by 
\begin{align}
    \theta_c^{(1)}&=  (\omega+\mu)\left( 1-\frac{1}{R_0}\right).
\end{align}

\subsection{Endemic level}\label{Sec_end_lev}
For $k=2$, we proved that the SIR\textsuperscript{(2)}S model without vaccination (with vaccination) has a unique endemic equilibrium whenever $R_0>1$ ($R_e^{(2)} > 1$). See Lemmas \ref{lem1} and \ref{lem2}.
For $k>2$, computing the endemic level from the SIR\textsuperscript{(k)}S model \eqref{SIRkS} implies finding feasible roots of a $k$'th degree polynomial function. This is numerically done for the parameter values in the (finely discretized) ranges in Table \ref{table:Params-values} as finding explicit formulae of the roots of high-degree polynomials is a challenging task. We obtain that the SIR\textsuperscript{(k)}S model admits a unique endemic equilibrium for the values in Table \ref{table:Params-values}. Those numerical simulations suggest the following conjectures.
\begin{conj}~
\begin{enumerate}
    \item The SIR\textsuperscript{(k)}S model \eqref{SIRkS} has a unique endemic equilibrium $E^{k,*} =(\bar s, \bar i, \bar r_0, \cdots, \bar r_{k-1})$ if and only if $R_0>1$.
    \item The SIR\textsuperscript{(k)}S model \eqref{SIRkS vaccination} has a unique endemic equilibrium $E_v^{k,*} = (\bar s^{\boldsymbol{\eta}}, \bar i^{\boldsymbol{\eta}}, \bar r_0^{\boldsymbol{\eta}}, \cdots, \bar r_{k-1}^{\boldsymbol{\eta}})$ if and only if $R_e^{(k)}>1$.
\end{enumerate}
\end{conj}
\subsection{Critical immunity level}
Fig. \ref{fig:Imm Lev Crit} shows the immunity level $\iota$ (as a function of $R_0$) at which individuals have to be vaccinated in order to reach herd immunity from the limiting SIR\textsuperscript{(k)}S models with linear and exponential decays of immunity. We recall that the classic SIRS model assume all individuals are either completely immune or completely susceptible, something which is not true in the models for gradual immunity waning. For parameter choices resembling the Covid-19 Delta strain ($R_0\approx 5$, $\omega^{-1}=12$ months and $\gamma^{-1}=7$ days), herd immunity will only be achieved if individuals are vaccinated before their immunity drops below $\iota\approx60\%$, according to the SIR\textsuperscript{(k)}S models with linear and exponential decays of immunity. This also means that individuals should get booster vaccines approximately every 6 months since their last vaccination/infection.
\begin{figure}[htbp]
\includegraphics[width=.7\linewidth]{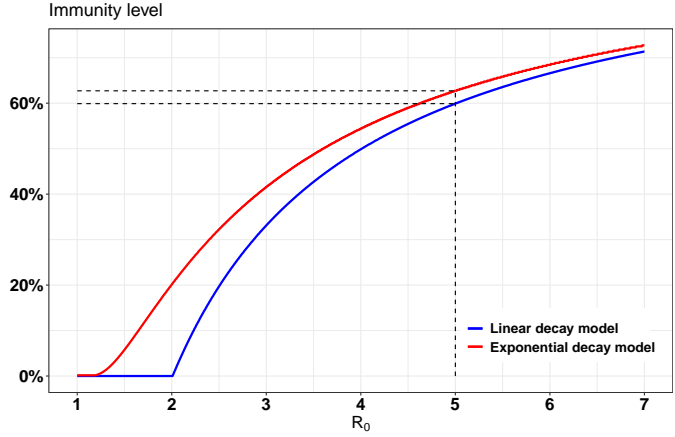} 
    \centering
  \caption{Immunity level $\iota$ at which individuals have to be vaccinated as a function of $R_0$,  with other parameter values from Table \ref{table:Params-values}. } 
  \label{fig:Imm Lev Crit}
\end{figure}

\subsection{Connection with ODE-PDE model}\label{Sec.PDE}
As $k\to\infty$, the number of states in the SIR\textsuperscript{(k)}S model increases and there is a continuity of immunity states. This limiting model can be described by an ODE-PDE system. Since we are interested in deterministic linear and exponential waning of immunity, knowing an individual's immunity level is equivalent to knowing the amount of time since his last recovery: time since recovery. The corresponding models can be formulated as follows.

\textbf{Linear decay model.} Assume a continuous linear decay of immunity and let $r(t,a)$ to be the fraction (density) of recovered individuals at time $t$ with age $a$ since recovery. For an infinitesimal time step $dt$, the individuals in $r(t,a)$ are those among $r(t- dt,a - dt) $ who will neither die nor get infected during the interval time $[t,t+dt]$, that is, we have for any $a>0$
\begin{align}\label{potentialPDE_disc}
r(t,a) =& r(t- dt,a - dt) \left( 1-\mu dt \right) \nonumber\\
& \times \left( 1-\beta \left(\left(\frac{\omega}{2}a-1\right){\boldsymbol 1}_{\{a<\frac{2}{\omega}\}}+1 \right)\, i(t) dt\right),\\
r(t,0) =& \gamma i(t).
\end{align}
Rearranging the \eqref{potentialPDE_disc} and sending $dt$ to $0$, it yields that
\begin{align}
    \dfrac{\partial r(t,a)}{\partial t} + \dfrac{\partial r(t,a)}{\partial a} = & - \beta \left(\left(\frac{\omega}{2}a-1\right){\boldsymbol 1}_{\{a<\frac{2}{\omega}\}}+1 \right)\, r(t,a) i(t)\nonumber\\
    &- \mu r(t,a).
\end{align}
Then the PDE-ODE model has the following form
\begin{align}\label{SIRkSPDEODELin}
&s'(t) = \mu - \beta s(t) i(t) - \mu s(t),  \nonumber\\
&i'(t) = \beta  \left(s(t)+\int_{2/\omega}^\infty\,r(t,\tau)\, d\tau \,\right) i(t) + \beta \int_0^{2/\omega}\, \frac{\omega}{2} \tau \,  r(t,\tau)\, d\tau \, i(t) \nonumber \\
  & \hspace{1cm} - (\gamma+\mu) i(t), \nonumber\\
&\dfrac{\partial r(t,a)}{\partial t} + \dfrac{\partial r(t,a)}{\partial a} = - \beta\left(\left(\frac{\omega}{2}a-1\right){\boldsymbol 1}_{\{a<\frac{2}{\omega}\}}+1 \right)\, r(t,a) i(t) \nonumber \\
 & \hspace{2.8cm} - \mu r(t,a),\quad a>0,
\end{align}
with the boundary condition $r(t,0) = \gamma i(t).$

\textbf{Exponential decay model.} Similarly to the previous paragraph, the corresponding PDE-ODE model when immunity wanes exponentially can be written as
\begin{align}\label{SIRkSPDEODEexp}
&s'(t) = \mu - \beta s(t) i(t) - \mu s(t), \nonumber  \\
&i'(t) = \beta  \left(s(t)+\int_{2/\omega}^\infty\,r(t,\tau)\, d\tau \,\right) i(t)  \nonumber \\
  & \hspace{1cm} + \beta \int_0^{\infty}\, \left( 1-e^{-\omega \tau}\right) \,  r(t,\tau)\, d\tau \, i(t) - (\gamma+\mu) i(t), \nonumber \\
&\dfrac{\partial r(t,a)}{\partial t} + \dfrac{\partial r(t,a)}{\partial a} = - \beta\left( 1-e^{-\omega a}\right)\, r(t,a) i(t) \nonumber \\
 & \hspace{2.8cm} - \mu r(t,a),\quad a>0,
\end{align}
with the boundary condition $r(t,0) = \gamma i(t).$

It has been shown in \citep{thieme2002endemic} that each of \eqref{SIRkSPDEODELin} and \eqref{SIRkSPDEODEexp} has only a disease free equilibrium (which is globally asymptotically stable) when the basic reproduction number $R_0\leq 1$, and a unique locally stable endemic equilibrium when $R_0>1$. We refer to the seminal works \citep{kermack1932contributions,kermack1933contributions} and the revisiting paper \citep{inaba2001kermack} for a general formulation in case both virgin and recovered individuals have varying susceptibility and infectives have variable infectivity. See also \citep{forien2022stochastic} where the authors consider the effects of previous infections on the susceptibility of partially susceptible individuals.

\bmhead{Acknowledgments}
The authors are grateful to the Swedish Research Council (grant 2020-04744) for financial support.

\noindent

\bibliography{main}

\newpage
\appendix
\section{Appendix}
\subsection{Case $k=2$: immunity waning in two steps}\label{appendix1}

We have the following results when immunity is lost in two steps.
\begin{lem}\label{lem1}
The SIR\textsuperscript{(2)}S model (3) in the main text has a unique endemic equilibrium $E^{2,*}=\left( \bar{s}, \bar{i}, \bar{r}_0,\bar{r}_{1}\right)$ if and only if $R_0>1$.
\end{lem}
Proof of Lemma \ref{lem1} is similar to the proof of the following lemma when the vaccination rates are equal to zero.
\begin{lem}\label{lem2}
The SIR\textsuperscript{(2)}S model (7) in the main text with vaccination has a unique endemic equilibrium $E_v^{2,*} =\left( \bar{s}^{\boldsymbol{\eta}}, \bar{i}^{\boldsymbol{\eta}}, \bar{r}_0^{\boldsymbol{\eta}},\bar{r}_{1}^{\boldsymbol{\eta}}\right)$ if and only if $R_e^{(2)}>1$.
\end{lem}
\begin{proof}
Solving the endemic equilibrium of equation (7) in the main paper for $k=2$ implies that the endemic level $\bar{i}^{\boldsymbol{\eta}}$ is the positive root of the following quadratic polynomial equation
\begin{align}\label{tosolve}
a x^2+ b x + \tilde{c} =0,
\end{align}
where
\begin{align*}
a =& \beta\frac{\mu}{c_2(1)} \left( c_2(1)+\gamma+\mu\right),\\
b =& \mu\frac{c_2(1)+\mu+\gamma}{c_2(1)} \left( \mu+\eta_s+2c_2(2) \right) - \beta \frac{c_2(1)+\mu}{c_2(1)} \left( \mu+2\dfrac{c_2(2)}{R_0}\right) \\
&+ \frac{\beta}{R_0} \left(2\dfrac{c_2(1)+\mu}{c_2(1)}(\mu+c_2(2)+\eta_1)-2\eta_1 \right)
,\\
\tilde{c} =& \left( (c_2(1)+\mu)(c_2(2)+\mu)+\eta_1\mu+\eta_s (c_2(1)+c_2(2)+\mu+\eta_1)\right) \\
&\times \frac{2\mu}{c_2(1) R_0}  \left(1-R_e^{(2)}\right).
\end{align*}
If $R_e^{(2)}>1$, we have $a\tilde{c}<0$ and then the equation \eqref{tosolve} has a unique positive root given by $\displaystyle\bar{i}^{\boldsymbol{\eta}}= (-b+\sqrt{b^2-4a\tilde{c}}) /(2a)$.
Furthermore, we have
\begin{align}
&\bar{s}^{\boldsymbol{\eta}} = \frac{\mu + c_2(2)/R_0}{ \left( \beta \bar{i}^{\boldsymbol{\eta}}+\mu+\eta_s\right)/2 +c_2(2)} ,\\
&\bar{r}_1^{\boldsymbol{\eta}} = 2 \left(  \frac{1}{R_0} - \bar{s}^{\boldsymbol{\eta}} \right),\\
&\bar{r}_0^{\boldsymbol{\eta}} = 1 - \bar{i}^{\boldsymbol{\eta}} - \bar{s}^{\boldsymbol{\eta}} - \bar{r}_1^{\boldsymbol{\eta}}.
\end{align}
Otherwise, it can be shown that all the coefficients $a,b$ and $\tilde{c}$ are non-negative and then the equation \eqref{tosolve} has no positive roots, that is, only the disease-free equilibrium exists when $R_e^{(2)}\leq 1$. 
\end{proof}
The following result gives the critical vaccine supply from the SI$\mbox{R}^{(2)}$S model with vaccination. 
\begin{prop}\label{proposition}~
\begin{enumerate}
\item If $\displaystyle 1<R_0<2(c_2(1)+c_2(2)+\mu)/c_2(1) , $ then, the critical vaccine supply is given by $$\theta_c^{(2)}=\dfrac{(c_2(1)+\mu)(c_2(2)+\mu)}{{c_2(1)+c_2(2)/2+\mu} } \left( 1-\dfrac{1}{R_0}\right). $$ 
\item If $\displaystyle R_0\geq 2(c_2(1)+c_2(2)+\mu)/c_2(1), $ then, the critical vaccine supply is given by $$\theta_c^{(2)}=\mu + \left(c_2(1) \dfrac{R_0}{2}-(c_2(1)+c_2(2)+\mu) + c_2(2)\right)\dfrac{2}{R_0}.$$
\end{enumerate}
\end{prop}
\begin{proof}
Let $(\eta_s,\eta_1)\in [0,\infty]^2$ such that  $R^{\boldsymbol{\eta}}=1$. Then, we have
\begin{align}
\theta^{(2)}\left( \eta_s,\eta_1\right)&= \eta_s \hat{s}^{\boldsymbol{\eta}}(\eta_s,\eta_1) + \eta_1 \hat{r}_1^{\boldsymbol{\eta}}(\eta_s,\eta_1)\\
&=\frac{\eta_s}{R_0}+ \left( \eta_1-\frac{\eta_s}{2} \right) \hat{r}_1^{\boldsymbol{\eta}}(\eta_s,\eta_1).
\end{align}
\begin{enumerate}
\item Assume that $\displaystyle 1<R_0<2(c_2(1)+c_2(2)+\mu)/c_2(1)  $ and set $$ \eta_s^\star:= \dfrac{(c_2(1)+\mu)(c_2(2)+\mu)}{(c_2(1)+c_2(2)+\mu) - c_2(1) R_0/2 } (R_0-1). $$
Then, we have
\begin{align}
&\theta^{(2)}\left( \eta_s,\eta_1\right)-\theta_c^{(2)}\left( \eta_s^\star, 0\right)\\
&=\frac{\eta_s}{R_0}+ \left( \eta_1-\frac{\eta_s}{2}\right) \hat{r}_1^{\boldsymbol{\eta}}\left( \eta_s,\eta_1\right)-\frac{\eta_s^\star}{R_0}+\frac{\eta_s^\star }{2} \hat{r}_1^{\boldsymbol{\eta}}\left(\eta_s^\star,0\right).
\end{align}
Rearranging the equality $R^{\boldsymbol{\eta}}=1$ allows to write $\eta_s$ in terms of $\eta_1$ as
\begin{align}
\eta_s = \dfrac{(c_2(1)+\mu)(c_2(2)+\mu)+\eta_1\mu}{c_2(1)+c_2(2)+\mu+\eta_1 - c_2(1) R_0/2}(R_0-1).
\end{align}
Therefore, a direct computation leads to
\begin{align}
&\theta^{(2)}\left( \eta_s,\eta_1\right)-\theta\left( \eta_s^\star, 0\right)\\
&= \dfrac{\eta_1c_2(1)(c_2(1)+\mu)(1-1/R_0)/2}{ \left( c_2(1)+c_2(2)+\mu - c_2(1)/2\right)\left( c_2(1)+c_2(2)+\mu+\eta_1 - c_2(1)/2\right)  }\nonumber\\
&\geq 0.
\end{align}
That is, the critical vaccine supply is given by 
\begin{align}
\theta_c^{(2)}=\theta^{(2)}\left( \eta_s^\star, 0\right)=\dfrac{(c_2(1)+\mu)(c_2(2)+\mu)}{c_2(1)+c_2(2)/2+\mu} \left( 1-\dfrac{1}{R_0}\right).
\end{align}
\item Now, assume that $\displaystyle R_0\geq 2(c_2(1)+c_2(2)+\mu)/c_2(1)  $ and set
$$ \eta_1^\star= c_2(1) R_0/2-(c_2(1)+c_2(2)+\mu)>0.$$
Similarly, one can show that
\begin{align}
\theta^{(2)}\left( \eta_s,\eta_1\right)-\theta^{(2)}\left( \infty,\eta_1^\star\right)= \dfrac{(c_2(1)+\mu) (c_2(2)+\mu  R_0/2 )}{ R_0(c_2(2)+(\mu+\eta_s)/2) } \geq 0.
\end{align}
and then the critical vaccine supply is given by 
\begin{align}
\theta_c^{(2)}=&\theta^{(2)}\left( \infty,\eta_1^\star\right) \\
=& \mu + (\eta_1^\star + c_2(2))\dfrac{2}{R_0}
\end{align}
where $\frac{2}{R_0}$ is the fraction of individuals in the class $r_j$ under the optimal vaccination strategy $( \infty,\eta_1^\star)$.
\end{enumerate}
\end{proof}
Fig. \ref{fig:sensit} plots the endemic level and the critical vaccine supply for the the SIR\textsuperscript{(2)}S models with linear and exponential decay of immunity, where it can be seen that the classic SIRS model underestimates both of the endemic level and the the critical vaccine supply. For a relatively small $k>1$, it should be clear that the long-term prevalence and the critical vaccine supply are affected by the choice of the transition rates between the immune states. When $k$ is large enough to fit the the linear and the exponential decays of immunity, both the corresponding long-term prevalence and the critical vaccine supply converge (Figs. 4 and 5 in the main text).

\begin{figure}[h!]
\centering
\begin{subfigure}{.7\linewidth}
\includegraphics[width=\linewidth]{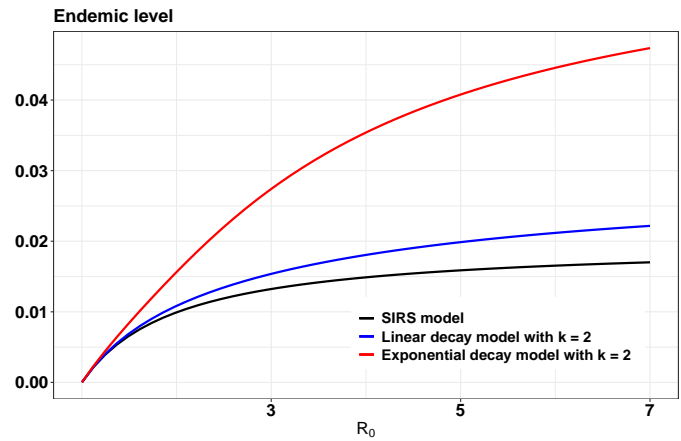}
\caption{}
\label{fig:endlevkequals2}
\end{subfigure}
\begin{subfigure}{.7\linewidth}
\includegraphics[width=\linewidth]{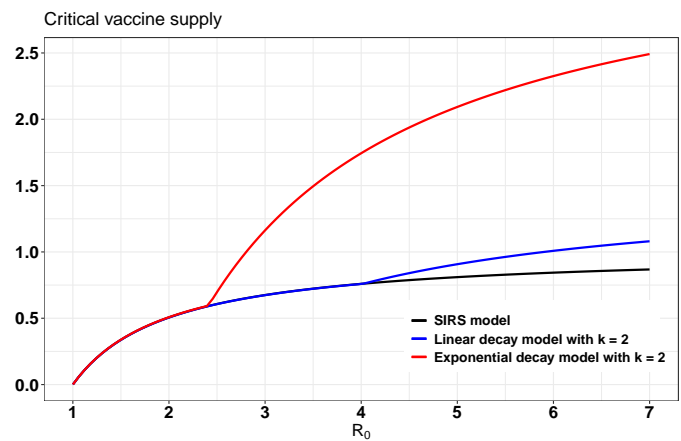}
\caption{}
\label{fig:VaccineSupplykequals2}
\end{subfigure}
\caption{(a) Endemic level and (b) critical vaccine supply from the classic SIRS model and the SIR\textsuperscript{(2)}S models with linear and exponential waning. All other parameters are set to their baseline values in Table 1 in the main text.}
\label{fig:sensit}
\end{figure}
\subsection{Sensitivity of the endemic level and the critical vaccine supply}
Figs. \ref{fig:EndLevelSens}--\ref{fig:vaxSens} plot the endemic level and the critical vaccine supply, as functions of the basic reproduction number $R_0 $, for the limiting SIR\textsuperscript{(k)}S epidemic models with linear and exponential waning of immunity. It can be seen that both the infectious period $\gamma^{-1} $ and the average immune period $\omega^{-1}$ strongly affect the long term prevalence for both modes of waning of immunity. We did not vary $\gamma^{-1} $ in Fig. \ref{fig:vaxSens} as the critical vaccine supply is only dependent on  $\gamma^{-1} $ through $R_0 $. It is also clear that varying the immunity duration from 6 months to 2 years can result in large reduction of the amount of vaccine needed to reach herd immunity.
\begin{figure}[htbp]
\centering
  \begin{subfigure}{.7\linewidth}
    \centering
    \includegraphics[width=\linewidth]{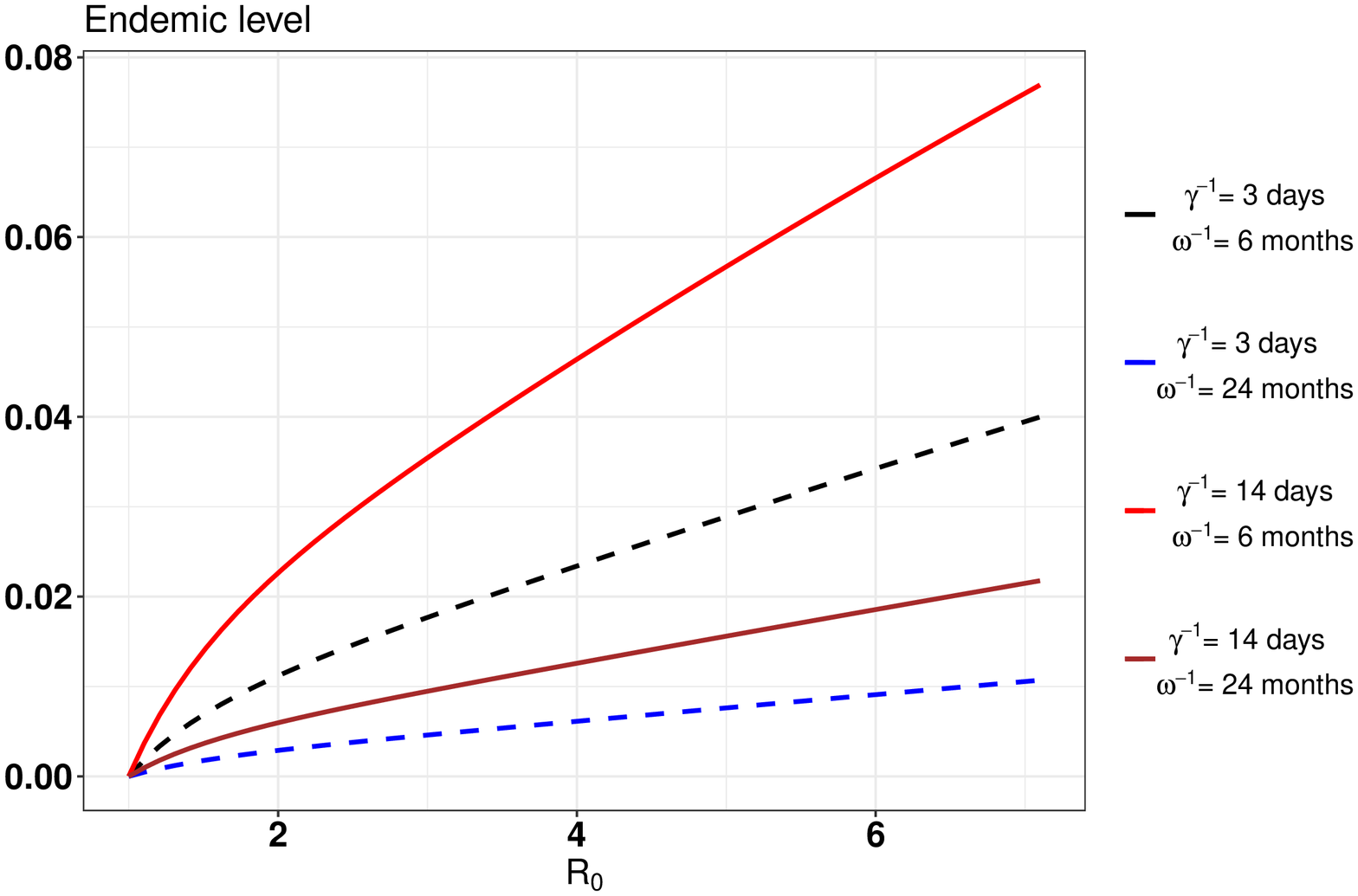}
    \caption{}
  \label{fig:EndLevelSensLin}
  \end{subfigure}
  \begin{subfigure}{.7\linewidth}
    \centering
    \includegraphics[width=\linewidth]{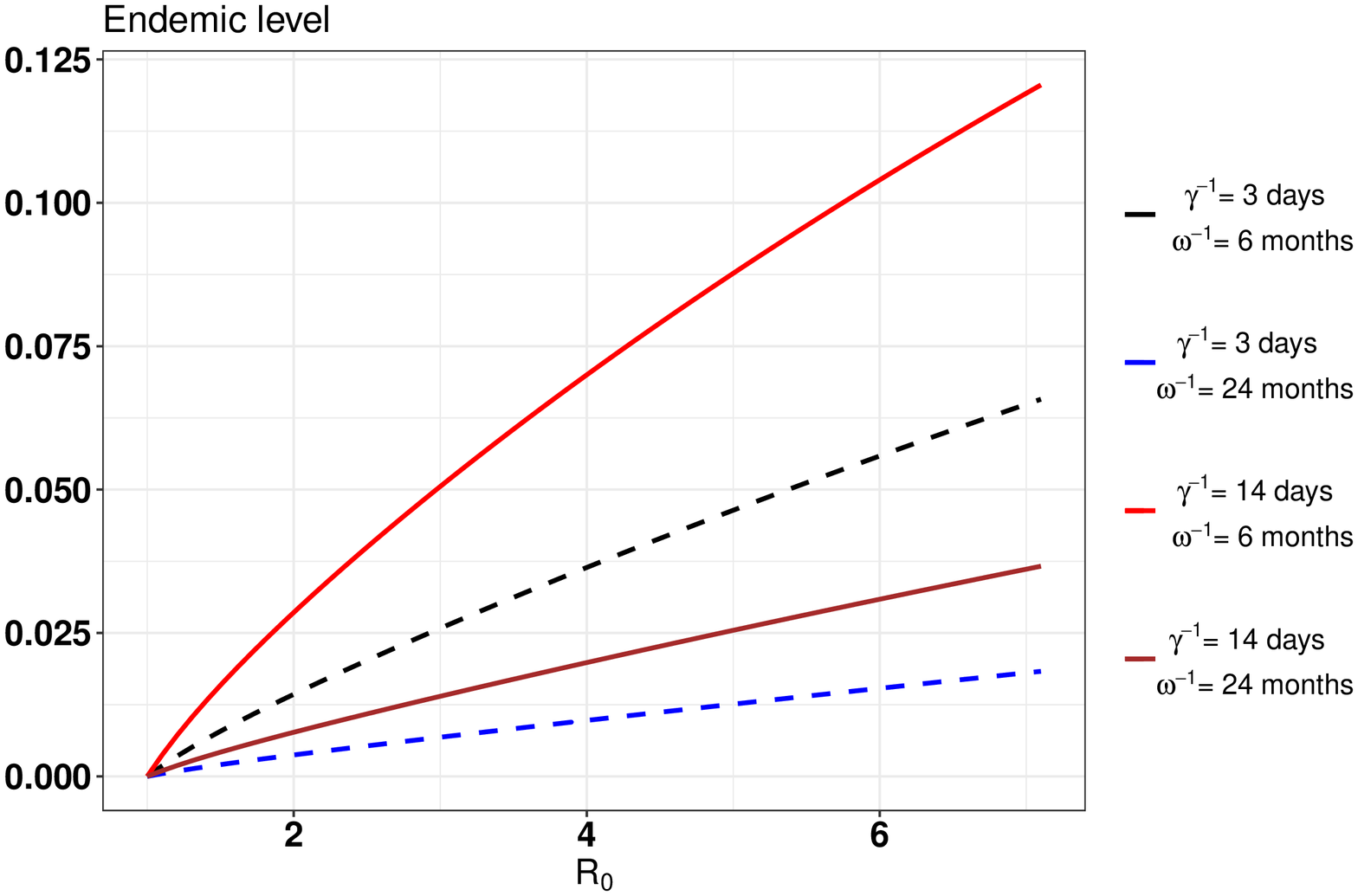}
    \caption{}
    \label{fig:EndLevelSensExp}
  \end{subfigure}
  \caption{Endemic levels from the limiting SIR\textsuperscript{(k)}S epidemic models, varying the mean infectious period and the average immunity duration from their baseline values. Solid lines correspond to the case where $\omega^{-1}=6$ months and dashed lines correspond to the case where $\omega^{-1}=24$ months. (a) Linear waning of immunity. (b) Exponential waning of immunity.}
  \label{fig:EndLevelSens}
\end{figure}
\begin{figure}[htbp]
\centering
  \begin{subfigure}{.7\linewidth}
    \centering
    \includegraphics[width=\linewidth]{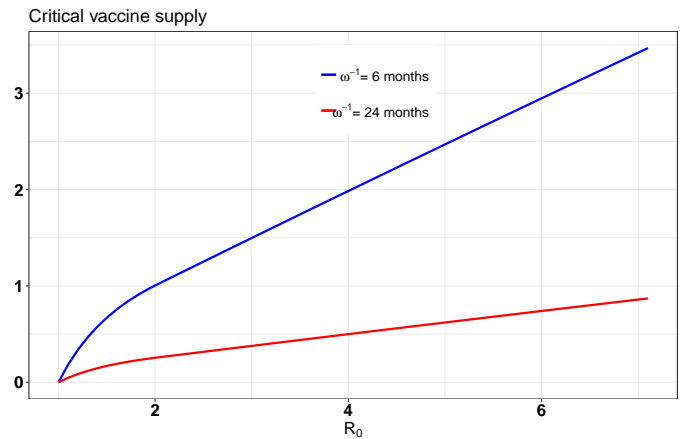}
    \caption{}
  \label{fig:vaxSensLin}
  \end{subfigure}
  \begin{subfigure}{.7\linewidth}
    \centering
    \includegraphics[width=\linewidth]{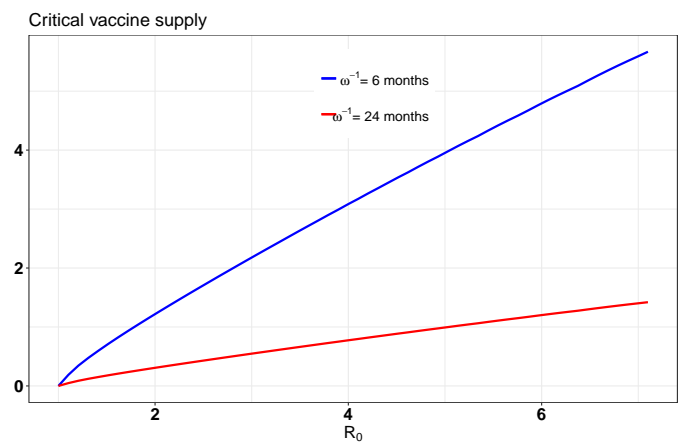}
    \caption{}
    \label{fig:vaxSensExp}
  \end{subfigure}
  \caption{Critical vaccine supply from the limiting SIR\textsuperscript{(k)}S epidemic model, varying the average immunity duration from its baseline value. (a) Linear waning of immunity. (b) Exponential waning of immunity.}
  \label{fig:vaxSens}
\end{figure}

\end{document}